\documentclass[letterpaper, 10 pt, conference]{ieeeconf}
\IEEEoverridecommandlockouts
\usepackage{fix-cm}
\usepackage[utf8]{inputenc}
\usepackage{cite}
\usepackage{caption}
\usepackage{subcaption}
\usepackage{amsmath,amssymb,amsfonts}
\usepackage{algorithmic}
\usepackage{algorithm} 
\usepackage{algorithmic}
\usepackage{booktabs}
\usepackage{graphicx}
\usepackage{subcaption}
\usepackage[textsize=footnotesize]{todonotes}
\newtheorem{theorem}{Theorem}
\newtheorem{remark}{Remark}
\newtheorem{corollary}{Corollary}

\newtheorem{definition}{Definition}
\newtheorem{lemma}{Lemma}

\newtheorem{example}{\indent Example}

\def\and{\textrm{ and}}

\newcommand{\Sj}[1]{S^{(#1)}(\prtition_#1,\alpha_{#1})}
\newcommand{\hj}[1]{h^{(#1)}(x,\prtition_#1,\alpha_{#1})}

\newcommand{\Domain}{\mathcal{D}}

\newcommand{\PWA}{\mathrm{PWA}}

\newcommand{\Vertex}{\mathcal{F}_0}

\newcommand{\prtition}{\mathcal{P}}

\newcommand{\Int}[1]{\mathrm{Int}\left(#1 \right)}

\usepackage{hyperref}
\hypersetup{
colorlinks=true,
linkbordercolor=orange,
citebordercolor=blue,
citecolor=blue,
filecolor=magenta,
urlcolor=magenta,
urlbordercolor=blue
}

\author{Pouya Samanipour and Hasan A. Poonawala
\thanks{Pouya Samanipour and Hasan A. Poonawala are with the Department of Mechanical and Aerospace Engineering, University of Kentucky, Lexington, USA
{\tt\small\{samanipour.pouya,hasan.poonawala\}@uky.edu}. The corresponding author is Hasan A.Poonawala. 
This work was supported by the National Science Foundation under Grant 2330794.}}
\begin{document}

\title{ReLU Barrier Functions for Nonlinear Systems with Constrained Control: A Union of Invariant Sets Approach}
\maketitle

\begin{abstract}
Certifying safety for nonlinear systems with polytopic input constraints is challenging because CBF synthesis must ensure control admissibility under saturation. 
We propose an approximation--verification pipeline that performs convex barrier synthesis on piecewise-affine (PWA) surrogates and certifies safety for the original nonlinear system via facet-wise verification. To reduce conservatism while preserving tractability, we use a two-slope Leaky ReLU surrogate for the extended class-$\mathcal{K}$ function $\alpha(\cdot)$ and combine multiple certificates using a \emph{Union of Invariant Sets (UIS)}.
Counterexamples are handled through local uncertainty updates.
Simulations on pendulum and cart-pole systems with input saturation show larger certified invariant sets than linear-$\alpha$ designs with tractable computation time.
\end{abstract}

\section{Introduction}
Safety is paramount in controlling dynamical systems, especially in robotics and autonomous driving where failures can be catastrophic. Formal safety guarantees typically certify that trajectories remain inside a safe region despite input constraints~\cite{hedesh2025delay}.

Barrier functions (BFs) and control barrier functions (CBFs) provide such certificates~\cite{ames2014control,ames2019control}: a CBF induces an invariant set by enforcing that, on its boundary, some admissible control keeps trajectories inside. Two practical obstacles are: (i) input constraints may eliminate feasible safety-preserving controls, and (ii) many methods assume prior knowledge of a safe/invariant set. When $h(x)$ must be synthesized, the choice of extended class-$\mathcal{K}$ function $\alpha(\cdot)$ further trades tractability for conservatism, with linear choices enabling convex programs but often shrinking the certified set.

The rise of machine learning~\cite{11152994} in control has enabled neural CBF approaches for safety-critical systems. However, these methods typically rely on nonconvex training with expensive verification~\cite{chang2019neural,dawson2022safe,zhang2023exact,zhao2021synthesizing}. SEEV~\cite{zhang2024seev} improves verification efficiency, but still depends on nonconvex training and does not provide convex synthesis under input saturation or a systematic mechanism to enlarge invariant sets. Data-driven Forward Invariant (FI) set estimation methods~\cite{11108033} can reduce modeling assumptions, but likewise do not yield convex synthesis under saturation. Soft-minimum barriers~\cite{10156245} improve feasibility under constraints by smoothing backup trajectories, yet require a known safe set and restrict $\alpha(\cdot)$, which can be conservative.

In our prior work~\cite{samanipour2024invariant}, we developed a convex framework for FI set construction for PWA systems with saturated inputs, but it (i) restricted $\alpha(\cdot)$ to be linear and (ii) applied only to PWA dynamics.

This paper synthesizes CBFs for nonlinear systems under input constraints while reducing conservatism without sacrificing convex synthesis. Leveraging the equivalence between ReLU networks and continuous PWA functions, we perform LP-based barrier synthesis on surrogate PWA models, then transfer guarantees to the true nonlinear system via facet-wise verification with local uncertainty updates. To systematically enlarge certified regions, we introduce a \emph{Union of Invariant Sets (UIS)} that combines multiple certificates through a max-composition (geometric intuition in Fig.~\ref{fig:UIS process}).

\noindent\textbf{Contributions.}
\begin{enumerate}
    \item \textbf{Leaky-ReLU surrogate for $\alpha(\cdot)$:}
    A two-slope design that avoids bilinear nonconvexity while reducing conservatism relative to linear $\alpha$.
    \item \textbf{Union of Invariant Sets (UIS):}
    A max-composition that certifies the union of multiple LP-based invariant sets (Fig.~\ref{fig:UIS process}), yielding a region never smaller than the best single-$\alpha$ certificate, with no extra optimization.
    \item \textbf{Approximation--verification pipeline:}
    Convex synthesis on surrogate PWA models combined with facet-wise verification and local uncertainty updates to ensure safety for the original nonlinear system under input constraints.
\end{enumerate}

Overall, the method preserves convex synthesis, reduces conservatism, and extends certificates from tractable surrogates to constrained nonlinear dynamics.

\begin{figure}
    \centering
    \includegraphics[width=0.65\linewidth]{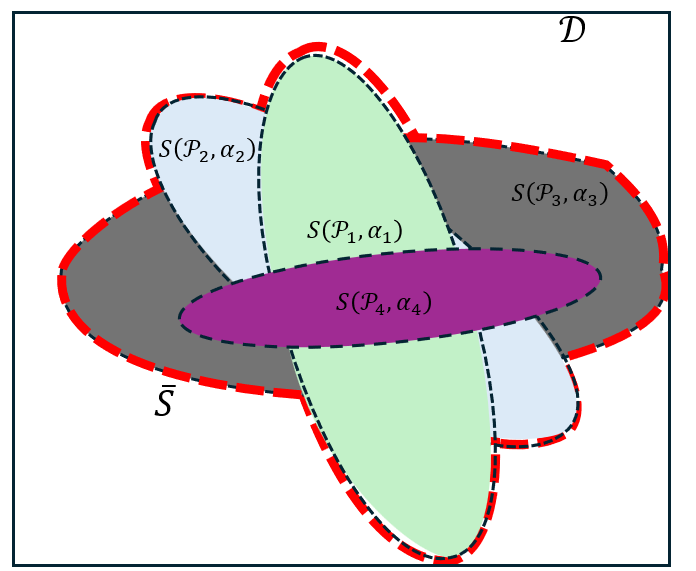}
    \caption{
Geometric interpretation of the Union of Invariant Sets (UIS). Each solid curve shows an invariant set obtained with a different linear $\alpha(x)$. Instead of selecting one candidate set, UIS certifies the union (red dashed line) via max-composition of the corresponding barrier functions. This typically yields a larger certified safe region and never smaller than the best single-$\alpha$ choice, without incurring additional optimization cost.}
\label{fig:UIS process}
\end{figure}


\section{Preliminaries}
We introduce the notation and concepts used in this paper.

\paragraph*{Notation.} 
For a set $S$, $\mathrm{conv}(S)$, $\mathrm{Int}(S)$, and $\partial S$ denote its convex hull, interior, and boundary, respectively. For a matrix $A$, $A^\top$ denotes its transpose.  

\subsection{Piecewise Affine Functions}
\label{sec:pwafun}
A piecewise affine (PWA) function is defined on a partition $\mathcal{P}=\{X_i\}_{i \in I(\mathcal{P})}$ of a domain $\Domain \subset \mathbb{R}^n$ by affine components $(A_i,a_i)$:
\begin{equation}\label{eq:definePWAfun}
f(x) = A_i x + a_i, \quad x \in X_i,
\end{equation}
where each cell $X_i$ is a bounded polytope of the form
\begin{equation}\label{eq:Hrep}
X_i = \{\,x \in \mathbb{R}^n \mid E_i x + e_i \succeq 0\,\}.
\end{equation}
Specifically, each cell $X_i$ in the partition can be represented as the convex hull of its set of vertices, $\Vertex(X_i)$, as follows:
\begin{equation}
X_i = \mathrm{conv}\{\Vertex(X_i)\}.
\end{equation}

\subsection{Forward Invariance and Barrier Functions}
We consider nonlinear control-affine dynamics
\begin{equation}\label{eq:system}
\dot{x} = f(x) + g(x)u, 
\quad x \in \Domain \subset \mathbb{R}^n,\;\; u \in \mathcal{U} \subset \mathbb{R}^m,
\end{equation}
where $f$ and $g$ are locally Lipschitz, and $\mathcal{U}=\{u:A_uu \leq c_u\}$ is a compact polytopic input set.

\begin{definition}[Forward invariance~\cite{ames2019control}]
A set $\mathcal{S}\subseteq \Domain$ is FI for~\eqref{eq:system} if $x(0)\in\mathcal{S}$ implies $x(t)\in\mathcal{S}$ for all $t\geq 0$.
\end{definition}

\begin{definition}[Control barrier function (CBF)~\cite{ames2019control}]
A continuously differentiable function $h:\Domain\to\mathbb{R}$ is a CBF for~\eqref{eq:system} if 
\begin{equation}\label{eq:safeset}
    \mathcal{S} = \{x \in \Domain \mid h(x) \geq 0\}
\end{equation}
is FI. A sufficient condition is the existence of an extended class-$\mathcal{K}_\infty$ function $\alpha$ such that
\begin{equation}\label{eq:CBF}
\inf_{u \in \mathcal{U}} \big[L_f h(x) + L_g h(x) u\big] 
\;\geq\; -\alpha(h(x)), \quad \forall x \in \Domain,
\end{equation}
where $L_f h$ and $L_g h$ denote Lie derivatives of $h$ along $f$ and $g$.
\end{definition}

\begin{definition}[Barrier function for closed-loop dynamics]
Given a feedback law $u=\kappa(x)$, the closed-loop system
\begin{equation}\label{eq:cl-system}
\dot{x} = f_{\mathrm{cl}}(x) := f(x) + g(x)\kappa(x)
\end{equation}
admits a barrier function $h(x)$ if the set $\mathcal{S}$ in~\eqref{eq:safeset} is FI whenever
\begin{equation}\label{eq:BF}
L_{f_{\mathrm{cl}}}h(x) \;\geq\; -\alpha(h(x)), \quad \forall x \in \Domain.
\end{equation}
\end{definition}

\section{Problem Description}
We aim to synthesize a CBF for the nonlinear control-affine dynamics~\eqref{eq:system}, subject to polytopic input constraints. The goal is to certify a forward-invariant set $\mathcal{S}$ defined in~\eqref{eq:safeset} that guarantees safety under all admissible controls.  

Obtaining such a CBF is nontrivial: for nonlinear dynamics, synthesizing a valid $h(x)$ under input constraints generally leads to a nonconvex optimization problem, which is computationally intractable in practice.  

To address this, we develop an approximation-verification pipeline that maintains \emph{convex optimization problems} while ensuring rigorous certification. The pipeline proceeds in five steps:  
\begin{enumerate}
    \item \textbf{Approximation:} Represent the closed-loop nonlinear dynamics with a piecewise affine (PWA) surrogate obtained from a ReLU neural network.  
    \item \textbf{Synthesis:} Construct a PWA barrier function for the surrogate dynamics via linear programming.
    \item \textbf{Enlargement:} Reduce conservatism by typically enlarging the certified invariant set using the Union of Invariant Sets (UIS) method, facilitated by a Leaky ReLU design for $\alpha(\cdot)$.
    \item \textbf{Verification:} Certify that the synthesized PWA barrier function satisfies the CBF condition~\eqref{eq:CBF} for the original nonlinear dynamics with polytopic input constraints, using a facet-wise check based on Farkas’ lemma.  
    \item \textbf{Local uncertainty update:} If verification fails, rather than relearning the nonlinear dynamics, locally update the surrogate model’s uncertainty and re-solve the convex program to obtain an adjusted barrier function.  
\end{enumerate}

This formulation bridges ReLU-based approximations with barrier-function design, provides a convex route to CBF synthesis for nonlinear systems, and ensures that failed certificates can be corrected through local surrogate uncertainty updates.

\section{\texorpdfstring{$\PWA$}{PWA} Barrier Functions for Surrogate Dynamics}
Direct synthesis of a CBF or certification of a FI set for the nonlinear dynamics~\eqref{eq:system} under polytopic input constraints is generally intractable. To address this, we fix a reference feedback law $u=\kappa(x)$ and approximate the resulting closed-loop dynamics~\eqref{eq:cl-system} on a bounded domain $\Domain$ by a piecewise affine (PWA) surrogate
\begin{equation}\label{eq:pwa dynamic}
    \hat f_{cl}(x) = \PWA(x), \quad x \in \Domain.
\end{equation}
Our primary objective in this section is to present an automated search algorithm for obtaining the certified invariant set of the surrogate PWA approximation. We achieve this through the optimization framework detailed below.

We build on the optimization framework from~\cite{samanipour2024invariant}, which certifies FI sets for PWA dynamics of the form~\eqref{eq:pwa dynamic}. Since this construction depends on the partition $\mathcal{P}$, we introduce the relevant index sets. The cells that intersect the domain boundary are
\begin{equation}
I_{\partial \Domain}(\mathcal{P}) = \{ i \in I(\mathcal{P}) \colon X_i \cap \partial \Domain \neq \emptyset \}.
\end{equation}
Vertices are classified as
\begin{align}
I_b(\mathcal{P}) &= \{ (i,k) \colon v_k \in \partial \Domain,\; i \in I_{\partial \Domain}(\mathcal{P}),\; v_k \in \Vertex(X_i) \}, \nonumber \\
I_{\mathrm{Int}}(\mathcal{P}) &= \{ (i,k) \colon v_k \notin \partial \Domain,\; i \in I(\mathcal{P}),\; v_k \in \Vertex(X_i) \}.
\end{align}
Here, $I_b(\mathcal{P})$ corresponds to boundary vertices, while $I_{\mathrm{Int}}(\mathcal{P})$ corresponds to interior vertices.

We parameterize the barrier function as a PWA function on the same partition:
\begin{equation}\label{eq: BF PWA}
    h_i(x) = s_i^T x + t_i, \qquad i \in I(\mathcal{P}),
\end{equation}
where $s_i \in \mathbb{R}^n$ and $t_i \in \mathbb{R}$.  
In cell $X_j$, we model local affine dynamics with bounded additive uncertainty,
\[
\dot{x} = A_j x + a_j + \delta x, \qquad \delta x \in \Delta_j,
\]
where $\Delta_j$ is a polytopic uncertainty set (challenges addressed in verification). The derivative of $h_i$ is
\begin{equation}\label{eq:BF derivative definition robust}
    \dot{h}_j(x) = s_j^T (A_j x + a_j + \delta x), \quad x \in \Int{X_j},\; \delta x \in \Delta_j.
\end{equation}
Since $s_j^T \delta x$ is linear in $\delta x$, it suffices to evaluate~\eqref{eq:BF derivative definition robust} at $\Vertex(\Delta_j)$.

The robust optimization problem is
\begin{subequations}\label{eq:opt UIS}
\begin{align}
&\min_{ s_i, t_i,\tau_{\mathrm{Int}_{i}},\tau_{b_{i}}}  \quad  \sum_{i=1}^{M}\tau_{b_{i}}+\sum_{i=1}^{N}\tau_{\mathrm{Int}_{i}},\label{eq:cost_function}\\
&\text{Subject to:} \nonumber\\
&h_i(v_k)-\tau_{b_i}\leq -\epsilon_1, \quad \forall (i,k) \in I_{b}(\mathcal{P}),\label{eq:NB}\\
&h_i(v_k)+\tau_{\mathrm{Int}_i}\geq\epsilon_2, \quad \forall (i,k) \in I_{\mathrm{Int}}(\mathcal{P}),\label{eq:PI} \\
&s_j^T(A_jv_k+a_j+\delta x)+\alpha h_i(v_k)\geq \epsilon_3, \nonumber\\
& \forall (i,k) \in I(\mathcal{P}), \; \forall \delta x \in \Vertex(\Delta_j),\label{eq:Nagumo}\\
&h_i(v_k)=h_j(v_k), \quad \forall v_k \in \Vertex(X_i) \cap \Vertex(X_j),\label{eq:continuity}\\
&\tau_{b_i},\tau_{\mathrm{Int}_{i}}\geq 0,\label{eq:PS}
\end{align}    
\end{subequations}
with tolerances $\epsilon_1,\epsilon_2,\epsilon_3>0$. In~\eqref{eq:cost_function}, $M$ and $N$ denote the number of boundary and interior vertices, respectively. Constraints~\eqref{eq:NB}–\eqref{eq:PI} regulate boundary and interior vertices, using slack variables to preserve feasibility. Constraint~\eqref{eq:Nagumo} enforces the barrier condition at all vertices of the uncertainty set, while~\eqref{eq:continuity} ensures continuity across cells.

\begin{remark}
The optimization~\eqref{eq:opt UIS} is always feasible, but a nonzero $\sum_{i}\tau_{b_i}$ indicates that the resulting set is not certified invariant. In this case, refinement of boundary (and possibly interior) cells is required, following the vector field refinement procedure of~\cite{10313502}.
\end{remark}

In~\eqref{eq:opt UIS}, the barrier function $h(x)$ is synthesized through optimization. Thus, the choice of the class-$\mathcal{K}_\infty$ function $\alpha(\cdot)$ not only controls the strictness of the barrier condition but also directly influences the resulting $h(x)$ and the size of the certified invariant set. Unlike standard CBF formulations where $h(x)$ is fixed and $\alpha(\cdot)$ acts as a tuning gain, here $\alpha(\cdot)$ fundamentally shapes the optimization landscape.  

In~\cite{samanipour2024invariant}, $\alpha(\cdot)$ was restricted to the linear form $\alpha(s)=\alpha s$, with $\alpha>0$ tuned via bisection search over $\{\alpha_1,\dots,\alpha_m\}$, retaining the largest feasible set. While effective, computing the volume of these sets in higher dimensions is expensive and often conservative, as nonlinear $\alpha(\cdot)$ can yield typically larger invariant sets.

This motivates exploring nonlinear $\alpha(\cdot)$, such as Leaky ReLU. To address the conservatism of linear $\alpha(\cdot)$, we introduce Leaky ReLU functions as a flexible choice, enabling less conservative invariant set estimation. Combined with the Union of Invariant Sets (UIS) method, this typically enlarges the certified region by merging multiple LP-derived invariant sets via max-composition.
\subsection{Leaky ReLU: A Practical Choice for \texorpdfstring{$\alpha(\cdot)$}{alpha(·)} in BF Design}
As shown in the prior section, the choice of $\alpha(\cdot)$ critically influences the certified FI set, since $h(x)$ is synthesized through optimization~\eqref{eq:opt UIS}. To avoid the conservatism of linear functions and the complexity of optimizing arbitrary $\mathcal{K}_\infty$ functions, we propose the Leaky ReLU as a practical substitute. It requires only two tunable parameters, yet approximates nonlinear $\mathcal{K}_\infty$ functions and aligns naturally with neural network architectures.
\begin{theorem}\label{th:leaky}
Let $S \subset \Domain \subseteq \mathbb{R}^n$ denote the super-level set~\eqref{eq:safeset} of a locally Lipschitz function $h:\Domain\!\to\!\mathbb{R}$ under the closed-loop dynamics~\eqref{eq:cl-system}. Assume $h(\Domain)\subset [h_{\min},h_{\max}]$ with $h_{\min}<0<h_{\max}$, where these bounds are derived from the problem domain.

The following are equivalent:
\begin{enumerate}
\item[(a)] There exists $\alpha\in\mathcal{K}_\infty$ such that the barrier condition~\eqref{eq:BF} holds on $\Domain$—using classical derivatives where $h$ is smooth and Clarke generalized derivatives otherwise—rendering $S$ FI. Moreover, $\alpha$ satisfies the one-sided slope bounds at $0$:
\begin{equation}\label{eq:onesided}
\limsup_{s\downarrow 0}\frac{\alpha(s)}{s}<\infty,
\qquad
\liminf_{s\uparrow 0}\frac{\alpha(s)}{s}>0.
\end{equation}

\item[(b)] There exist $0<\alpha_1\le \alpha_m<\infty$ and the extended $\mathcal{K}_\infty$ function
\begin{equation}\label{eq:Leaky relu}
\overline{\alpha}(s) = 
\begin{cases} 
\alpha_m s, & s \ge 0, \\
\alpha_1 s, & s < 0,
\end{cases}
\end{equation}
such that the barrier condition~\eqref{eq:BF} holds on $\Domain$, rendering $S$ FI.
\end{enumerate}
\end{theorem}

\begin{proof}
\textbf{(b) $\Rightarrow$ (a).}
The function $\overline{\alpha}$ in~\eqref{eq:Leaky relu} is continuous, strictly increasing, radially unbounded, and satisfies $\overline{\alpha}(0)=0$, hence $\overline{\alpha}\in\mathcal{K}_\infty$. For locally Lipschitz $h$, let $\partial^C h(x)$ denote the Clarke generalized gradient. The CBF condition~\eqref{eq:BF} in the nonsmooth setting is
\[
\sup_{\zeta \in \partial^C h(x)} \, \zeta^\top f_{\mathrm{cl}}(x) \;\ge\; -\,\overline{\alpha}\big(h(x)\big), \quad \forall x\in\Domain,
\]
which aligns with the classical condition where $h$ is differentiable. Standard nonsmooth CBF results (e.g.,~\cite{glotfelter2017nonsmooth}) confirm this ensures forward invariance of $S=\{x:h(x)\ge 0\}$. Thus (a) holds.

\smallskip
\textbf{(a) $\Rightarrow$ (b).}
Let $\alpha$ satisfy~\eqref{eq:onesided} and the (smooth/nonsmooth) CBF inequality in~\eqref{eq:BF}. Define
\[
\alpha_m := \sup_{s\in(0,h_{\max}]}\frac{\alpha(s)}{s}, 
\qquad
\alpha_1 := \inf_{s\in[h_{\min},0)}\frac{\alpha(s)}{s}.
\]
By compactness of $[h_{\min},h_{\max}]$ and the one-sided slope bounds~\eqref{eq:onesided}, $0<\alpha_1\le \alpha_m<\infty$. For all $s\in[h_{\min},h_{\max}]$,
\[
\alpha(s)\;\le\;
\begin{cases}
\alpha_m s,& s\ge 0,\\
\alpha_1 s,& s<0,
\end{cases}
=:\overline{\alpha}(s),
\quad \overline{\alpha}(0)=0.
\]
Thus, for every $x\in \Domain$,
\[
\sup_{\zeta \in \partial^C h(x)} \, \zeta^\top f_{\mathrm{cl}}(x) 
\;\ge\; -\,\alpha\big(h(x)\big) 
\;\ge\; -\,\overline{\alpha}\big(h(x)\big),
\]
showing the CBF condition holds with $\overline{\alpha}$, ensuring $S$’s invariance. Thus (b) holds.
\end{proof}

\begin{remark}
Theorem~\ref{th:leaky} applies to locally Lipschitz $h$, covering both smooth and nonsmooth barrier functions. In the nonsmooth case, the derivative in~\eqref{eq:BF} uses Clarke gradients, aligning with~\cite{glotfelter2017nonsmooth}, and supports PWA functions from~\eqref{eq:opt UIS} as well as smooth cases.
\end{remark}
\begin{remark}
Theorem~\ref{th:leaky} is essential as $h(x)$, synthesized via~\eqref{eq:opt UIS}, is unknown beforehand. In our optimization, $\alpha(\cdot)$ must uphold the barrier constraint~\eqref{eq:BF} across all points in $\Domain$, unlike standard CBFs that verify it only on $\partial S$, ensuring $S$’s invariance.
\end{remark}

This result enables certification with a two-parameter Leaky ReLU. Any invariant set certifiable by $\alpha\in\mathcal{K}_\infty$ is also certifiable by $\overline{\alpha}$. We adopt this as a tractable, NN-compatible surrogate for $\alpha(\cdot)$, simplifying optimization while preserving certifiability.

\subsection{Union of Invariant Sets (UIS)}\label{sec:UIS}
Building on the Leaky ReLU design for $\alpha(\cdot)$, we introduce the Union of Invariant Sets (UIS) method. 
Directly encoding a Leaky ReLU $\alpha(\cdot)$ into~\eqref{eq:opt UIS} makes the barrier constraint nonconvex, typically requiring a mixed-integer program that undermines convexity. 
Instead, UIS solves~\eqref{eq:opt UIS} for multiple linear $\alpha(x) = \gamma x$, each as a convex program, and unites the resulting invariant sets. 
The key innovation is max-composition of the associated barrier functions, certifying the \emph{union} as forward-invariant. 
Unlike our prior $\alpha$-sweep, which tested various $\gamma$ values, selected the largest single invariant set, and discarded other results—wasting computation—UIS harnesses all LP solutions to yield a single, enlarged certified region without additional solves or costly set comparisons.

\paragraph*{Notation.}
Because the solution of~\eqref{eq:opt UIS} depends on both the partition and $\alpha(\cdot)$, we write the certified set and barrier as
$S^{(i)}:=S^{(i)}(\mathcal{P}_i,\alpha_i)$ and $h^{(i)}:=h^{(i)}(x,\mathcal{P}_i,\alpha_i)$, with $\alpha_i(x)=\alpha_i x$ unless otherwise stated.
\begin{lemma}\label{lemma:UIS}
Consider the dynamical system given by~\eqref{eq:pwa dynamic} with an equilibrium at the origin, defined on a compact set $\Domain$.
Let $\underline{\alpha} = [\alpha_{min},\dots, \alpha_{max}]$ be a set of $m$ parameters ordered in ascending order. 
Suppose that the corresponding invariant sets, obtained through the optimization problem~\eqref{eq:opt UIS} are denoted by $\Sj{i}$ for $i=1,\dots,m$ with $\alpha(x)=\alpha_ix$.
Then, the following results hold:
\begin{enumerate}
    \item There exists an invariant set $\overline{S}$ with respect to the dynamical system~\eqref{eq:pwa dynamic}, where \( \overline{S} \) is given by:
    \begin{equation}\label{eq:union of of inv}
    \overline{S} = \bigcup_{i=1}^m \Sj{i}.    
    \end{equation}
    \item The set $\overline{S}$ is rendered \emph{asymptotically stable} in $\Domain$ with the following BF and class-$\mathcal{K}_{\infty}$ function:
    \begin{equation}\label{eq:h UIS}
    \overline{h}(x, \overline{\prtition}, \overline{\alpha}) = \max_{i} \{ \hj{i} \},
    \end{equation}
    where $\overline{\alpha}(x)$ is defined as:
    \begin{align}
        \overline{\alpha}(x) = \alpha_{max} \sigma_{\left(\frac{\alpha_{min}}{\alpha_{max}}\right)}(x),
    \end{align}
    and $\sigma_{\left(\frac{\alpha_{min}}{\alpha_{max}}\right)}(x)$ denotes the Leaky ReLU function as defined in~\eqref{eq:Leaky relu}. The partition $\overline{\prtition}$ is the product partition defined as:
    \begin{equation}
        \overline{\mathcal{P}} = \mathcal{P}_1 \times \mathcal{P}_2 \times \dots \times \mathcal{P}_m, 
    \end{equation}
    where:
    \begin{equation}
     \prtition^*=\prtition_1\times\prtition_2=\{X_k\}_{k\in I(\prtition^*)},\nonumber   
    \end{equation}
    \begin{equation}
    X_k=\{X_i\cap X_j:i\in I(\prtition_1) ,j\in I(\prtition_2),dim(X_k)=n\}.\nonumber    
    \end{equation}
\end{enumerate}
\end{lemma}

\begin{proof}
(1) If $x_0\in\overline{S}$, then $x_0\in S^{(k)}$ for some $k$. Since $S^{(k)}$ is invariant for~\eqref{eq:pwa dynamic}, the trajectory remains in $S^{(k)}\subseteq \overline{S}$, proving invariance.

(2) By construction $\overline{h}=\max_i h^{(i)}(x,\prtition_i,\alpha_i)$, so $\overline{S}=\{x:\overline{h}(x)\ge 0\}$, and $\overline{h}$ is locally Lipschitz and PWA on $\overline{\mathcal{P}}$. 
At differentiable points $x^*$ where $\overline{h}(x^*)=h^{(k)}(x^*,\prtition_k,\alpha_k)$,
\[
\dot{\overline{h}}(x^*)=\dot{h}^{(k)}(x^*,\prtition_k,\alpha_k).
\]
If $\overline{h}(x^*)>0$, then $\alpha_{\max}\overline{h}(x^*)\ge \alpha_k h^{(k)}(x^*,\prtition_k,\alpha_k)$ and feasibility of $h^{(k)}()$ gives
\begin{align*}
&\dot{\overline{h}}(x^*)+\alpha_{\max}\overline{h}(x^*)\ \ge \\ &\dot{h}^{(k)}(x^*,\prtition_k,\alpha_k)+\alpha_k h^{(k)}(x^*,\prtition_k,\alpha_k)\ \ge\ 0.
\end{align*}
If $\overline{h}(x^*)<0$, then $\overline{\alpha}(\overline{h}(x^*))=\alpha_{\min}\overline{h}(x)\ge \alpha_k h^{(k)}(x^*,\prtition_k,\alpha_k)$, so
\[
\dot{\overline{h}}(x^*)+\overline{\alpha}(\overline{h}(x*))\ \ge\ 0.
\]

\noindent\textbf{Nondifferentiable points.}
Since $\overline{h}(x)=\max_i h^{(i)}(x,\prtition_i,\alpha_i)$ and each $h^{(i)}(x,\prtition_i,\alpha_i)$ is locally Lipschitz, Clarke calculus yields
\begin{align*}
&\partial^{C}\overline{h}(x)=\\&\mathrm{co}\Big\{\partial^{C}h^{(i)}(x,\prtition_i,\alpha_i):\ i\in\arg\max_{j} h^{(j)}(x,\prtition_j,\alpha_j)\Big\}.
\end{align*}
Hence
$$
\begin{aligned}
&\sup_{\zeta\in \partial^{C}\overline{h}(x)}\ \zeta^\top \PWA(x)\\
&= \max_{i\in \arg\max_{j} h^{(j)}(x,\prtition_j,\alpha_j)}\ \sup_{\zeta\in \partial^{C}h^{(i)}(x,\prtition_i,\alpha_i)}\ \zeta^\top \PWA(x)\\
&\ge -\,\overline{\alpha}\!\big(\overline{h}(x)\big),
\end{aligned}
$$
since each active $h^{(i)}$ satisfies the BF inequality and $\overline{\alpha}$ dominates the corresponding linear slope.
Therefore, the BF condition holds everywhere for $\overline{h}$ with $\overline{\alpha}$, implying forward invariance and asymptotic stability of $\overline{S}$ in $\Domain$.
\end{proof}
\begin{remark}
UIS requires no additional optimization beyond the baseline multi-$\alpha$ synthesis. 
It simply reuses the LP solutions already obtained for different linear gains and performs a max-composition step, yielding the certified union without extra solve cost.
\end{remark}

\begin{corollary}\label{cor:UIS}
Let $S'=\arg\max_{i} \mathrm{vol}(S^{(i)})$ among sets obtained from Algorithm~\ref{alg:union of invariant sets} with linear $\alpha(x)=\alpha_i x$. 
Then $S'\subseteq \overline{S}$, i.e., UIS never returns a smaller certified set than the best single linear choice.
\end{corollary}
\begin{proof}
    The result follows directly from the construction of the UIS algorithm.
\end{proof}

Corollary~\ref{cor:UIS} shows that UIS does not always yield a strictly larger invariant set (e.g., $S(\mathcal{P}_4,\alpha_4)$ in Fig.~\ref{fig:UIS process} may add no new points). Importantly, since our baseline already solves the multi-$\alpha$ LP sweep, UIS incurs \emph{no additional optimization cost} while still enabling potentially larger certified sets. In practice, we select an initial $\alpha_{0}$ uniformly at random from $\underline{\alpha}$, solve for $S(\prtition_0,\alpha_0)$, and then evaluate nearby values of $\alpha$.
Note that Algorithm~\ref{alg:union of invariant sets} as stated has no built-in termination condition and can in principle run indefinitely.
\begin{algorithm}[tb]
\begin{algorithmic}[1]
\REQUIRE $\PWA(x)$ dynamics~\eqref{eq:pwa dynamic}, slope grid $\underline{\alpha}=\{\alpha_1,\dots,\alpha_m\}$.
\FOR{$j=1,2,\dots,m$}
    \STATE Set $\alpha(s)=\alpha_j s$.
    \STATE Solve~\eqref{eq:cost_function}--\eqref{eq:PS} for~\eqref{eq:pwa dynamic} with $\alpha(s)$ to obtain 
    $h^{(j)}(x,\mathcal{P}_j,\alpha_j)$ and the certified set $S^{(j)}(\mathcal{P}_j,\alpha_j)$.
    \WHILE{$\sum_{i=1}^{M}\tau_{b_i}\neq 0\ \ \textbf{or}\ \ \sum_{i=1}^{N}\tau_{\mathrm{Int}_i}\neq 0$}
        \STATE Refine cells with nonzero slacks as in~\cite{10313502}.
        \STATE Re-solve~\eqref{eq:cost_function}--\eqref{eq:PS} on the refined partition to update 
        $h^{(j)}(x,\mathcal{P}_j,\alpha_j)$ and $S^{(j)}(\mathcal{P}_j,\alpha_j)$.
    \ENDWHILE
\ENDFOR
\STATE Form the UIS barrier and set:
\begin{align*}
    &\overline{h}(x,\overline{\mathcal{P}},\overline{\alpha})=\max_{j} h^{(j)}(x,\mathcal{P}_j,\alpha_j)\\
    &\overline{S}=\bigcup_{j=1}^{m} S(\mathcal{P}_j,\alpha_j)
\end{align*}
\RETURN $\overline{S}$ and $\overline{h}(x,\overline{\mathcal{P}},\overline{\alpha})$.
\end{algorithmic}
\caption{Union of Invariant Sets (UIS) for $\PWA$ dynamics~\eqref{eq:pwa dynamic}}
\label{alg:union of invariant sets}
\end{algorithm}

\section{Constructing CBFs for Nonlinear Dynamics}

The UIS framework estimates FI sets for surrogate PWA dynamics~\eqref{eq:pwa dynamic}, but this does not imply invariance for the original nonlinear system~\eqref{eq:system}. A natural idea is to incorporate model mismatch as additive uncertainty in each polytopic region. However, doing so from the outset poses two challenges: 
(i) it significantly increases the size of the LP~\eqref{eq:opt UIS}, and 
(ii) constructing accurate uncertainty bounds $\Delta_j$ for every region is nontrivial and often conservative.

We instead adopt a scalable \emph{solve–verify–update} pipeline. First, we synthesize a barrier function for the surrogate PWA model with no uncertainty. This step uses a fixed nominal feedback controller, so the resulting barrier certifies a forward-invariant set for the induced surrogate closed-loop dynamics.

Next, we verify whether the candidate barrier satisfies the CBF condition for the original nonlinear system with input constraints. Unlike synthesis, this step checks whether \emph{any admissible control} exists that can keep trajectories within the estimated safe set. Verification is performed via Farkas' lemma on selected boundary regions.

Importantly, because the candidate barrier is piecewise affine, its zero-level set can be decomposed into a finite number of \emph{facet patches}—each corresponding to a local active region of the max-composed function. We exploit this structure to check the CBF condition only on these boundary patches, not the full domain. This patch-wise verification is exact and highly efficient.

If a violation is found, we extract a counterexample state and compute the deviation between the true dynamics and the surrogate model at that point. This yields a \emph{cell-local uncertainty set} $\Delta_j$, which is added only to the affected region. We then re-solve the LP globally, keeping the rest of the model unchanged. This selective update avoids conservative global uncertainty while preserving tractability.

\paragraph*{Facet index set.}
Let
\[
\overline{h}(x,\overline{\mathcal{P}},\overline{\alpha})
=\max_{k\in I_h} h_k(x,\mathcal{P}_k,\alpha_k),\qquad I_h=\{1,2,\dots,m\},
\]
and, let $\overline{\mathcal{P}}=\{X_j\}$ denote the fixed (refined) partition of $\overline{h}$. Collect the boundary cell–active piece pairs
\begin{align*}
  \mathcal{I}_F \;:=\;\bigl\{(j,k)\;:\;& \exists\,x\in X_j \text{ s.t. } \overline{h}(x,\overline{\mathcal{P}},\overline{\alpha})=0,\\
  & \overline{h}(x,\overline{\mathcal{P}},\overline{\alpha})=h_k(x,\mathcal{P}_k,\alpha_k)\bigr\}.
\end{align*}
For each $(j,k)\in\mathcal{I}_F$, define the facet patch
\begin{align*}
F_{j,k} \;:=\;&\bigl\{\, x\in X_j:\; h_k(x,\mathcal{P}_k,\alpha_k)=0,\\
& \ \ h_k(x,\mathcal{P}_k,\alpha_k)\ge h_\ell(x,\mathcal{P}_\ell,\alpha_\ell)\ \forall \ell\in I_h \bigr\}.
\end{align*}

\paragraph*{Local verification }
Consider the input constraint set $U=\{u:A_uu\le c_u\}$.  
On each facet patch $F_{j,k}$, we check the CBF condition~\eqref{eq:CBF} by eliminating $u$ with Farkas’ lemma for~\eqref{eq:system}. For every $(j,k)\in\mathcal{I}_F$, solve
\begin{equation}
\label{eq:facet-verify-pairs}
\begin{aligned}
\Phi_{j,k}^\star:=\min_{x,\ \lambda\ge 0}\quad & s_k^\top f(x)+\lambda^\top c_u\\
\text{s.t.}\quad & x\in F_{j,k},\\
& \lambda^\top A_u=s_k^\top g(x),
\end{aligned}
\end{equation}
where $s_k:=\nabla h_k$ is constant on $X_j$. If $\Phi_{j,k}^\star<0$ then a counterexample exists on $F_{j,k}$—i.e., there is an $x\in F_{j,k}$ for which \emph{no} $u\in U$ can satisfy~\eqref{eq:CBF}; otherwise the CBF constraint holds on that patch. 
This guarantee is independent of the nominal law $\kappa(x)$ used during synthesis, ensuring the safe set is physically realizable even if the practical control law differs from the nominal one.
In implementation, we apply a standard CBF-QP safety filter that minimally modifies $\kappa(x)$, and feasibility of this filter on each verified facet patch follows from~\eqref{eq:facet-verify-pairs}.
This elimination-based “exact verification” mirrors the dual/support-function formulations used in~\cite{zhang2023exact}.


\paragraph*{Practical note.}
Problem~\eqref{eq:facet-verify-pairs} has decision variables $(x,\lambda)$ only, with $x\in\mathbb{R}^n$ and $\lambda\in\mathbb{R}^p$, where $p$ is the number of inequality rows in $A$ for $U=\{u:A_uu\le c_u\}$. Thus the program dimension is $n+p$; the control $u$ is eliminated via Farkas’ lemma. The set constraint $x\in F_{j,k}\subseteq X_j$ is polyhedral (hence cheap), while the coupling $\lambda^\top A_u=s_k^\top g(x)$ and the term $s_k^\top f(x)$ are smooth in $x$. Each per–$(j,k)$ instance is local and small, and can be solved reliably with off-the-shelf nonlinear solvers (e.g., \texttt{scipy.optimize}).

\subsection{Cell-local robustification via uncertainty update}
If verification produces a counterexample $x^\ast\in F_{j,k}$, set
$e^\ast:=f(x^\ast)-\hat f(x^\ast),$
\[
\Delta_j:=
\begin{cases}
\{\delta x:\ |\delta x|\le |e^\ast|\}, & \Phi_{j,k}^\star<0,\\
\{0\}, & \text{otherwise,}
\end{cases}
\]
where $|\cdot|$ is taken componentwise (box uncertainty). Using the counterexample, we obtain a tighter \emph{local} uncertainty only in cell $X_j$, avoiding an overly conservative global bound. 

\subsection{Search algorithm}

The complete procedure is summarized in Algorithm~\ref{alg:uis-verify-robustify-refine}. We begin by applying UIS to synthesize a candidate barrier function on the surrogate dynamics without uncertainty. This is followed by patch-wise verification on the original nonlinear system. If a violation is detected, we compute a cell-local uncertainty $\Delta_j$, re-solve the LP with this localized robustification, and repeat. The loop terminates when all patches pass verification, yielding a valid CBF for the nonlinear system.

\paragraph*{Scalability.}
Both the UIS solves (per cell) and the facet verifications (per $(j,k)$) are independently parallelizable (no cross-coupling), enabling efficient multicore/GPU execution; optional partition refinement can be applied locally when needed.

\begin{algorithm}[t]
\caption{UIS $\to$ Nonlinear Verify $\to$ Cell-Local Robustify \& Refine (no-failure)}
\label{alg:uis-verify-robustify-refine}
\begin{algorithmic}[1]
\REQUIRE Nonlinear dynamic~\eqref{eq:cl-system}, fixed partition $\mathcal{P}^{\ast}=\{X_j\}$, nominal PWA model $\hat f$, input set $U=\{u:A_uu\le c_u\}$, tolerances $\varepsilon,\varepsilon_3$
\STATE \textbf{Start from UIS (no uncertainty):} Solve \eqref{eq:opt UIS} with $\Delta_j=\emptyset\ \forall j$ to obtain $\overline{h}(x,\overline{\mathcal{P}},\overline{\alpha})=\max_i h^{(i)}(x,\mathcal{P}_i,\alpha_i)$.
\STATE \textbf{Verify on nonlinear facets:} For each $(j,k)\in\mathcal{I}_F$ (facet $F_{j,k}$), compute $\Phi_{j,k}^\star$ via \eqref{eq:facet-verify-pairs} \emph{using the original dynamics} $f(x),g(x)$ and input set $U$.
\WHILE{there exists $(j,k)\in\mathcal{I}_F$ with $\Phi_{j,k}^\star<0$ (witness $x^{\ast}\in F_{j,k}$)}
     \STATE \textbf{Update $\Delta_j$ \& re-solve:}
     For each counterexample $x^{\ast}$, let $e^{\ast}=f(x^{\ast})-\hat f(x^{\ast})$ and update
     $\Delta_j \leftarrow \operatorname{conv}\!\big(\Delta_j \cup \{\pm e^{\ast}\}\big)$;
     re-solve \eqref{eq:opt UIS}.
  \WHILE{$\sum_{i=1}^{N} \tau_{b_i} \neq 0$}
      \FOR{each $i$ with $\tau_{b_i}\neq 0$ \textbf{and} $\tau_{\mathrm{Int}_i}\neq 0$}
          \STATE Refine cells as in~\cite{10313502}.
      \ENDFOR
      \STATE Re-solve the same LP~\eqref{eq:opt UIS} (robust only on cells with $\Delta_j\neq\emptyset$).
  \ENDWHILE
  \STATE Update $\overline{h}(x,\overline{\mathcal{P}},\overline{\alpha})=\max_i h_i(x,\mathcal{P}_i,\alpha_i)$ from the LP solution.
  \STATE \textbf{Re-verify:} Recompute $\Phi_{j,k}^\star$ via~\eqref{eq:facet-verify-pairs} \emph{using the original dynamics} $f(x),g(x)$ and $U$ on all facets.
\ENDWHILE
\STATE \textbf{Return} $\overline{h}(x,\overline{\mathcal{P}},\overline{\alpha})$ as the certified CBF for \eqref{eq:system}; the safe set is $S=\{x:\overline{h}(x,\overline{\mathcal{P}},\overline{\alpha})\ge 0\}$.
\end{algorithmic}
\end{algorithm}

\section{Results and Simulations}
All computations are performed in Python 3.11 on a 2.1 GHz CPU with 8 GB RAM. A tolerance of $10^{-6}$ determines nonzero values, and $\epsilon_1=\epsilon_2=\epsilon_3=10^{-4}$ in the examples.
\begin{example}[Inverted Pendulum~\cite{samanipour2024invariant}]\label{example:IP}
The inverted pendulum system can be modeled as follows:
\begin{align}
\dot{x}_1 &= x_2 \nonumber, \quad\dot{x}_2 = \sin(x_1) + u \nonumber
\end{align}
where $x_1$ represents the pendulum's angle, and $x_2$ is its angular velocity. 
We apply a nominal control law $u = -3x_1 - 3x_2$, subject to saturation limits between $-1.5$ and $1.5$. 
The system's domain is $\Domain = \pi - ||x||_\infty \geq 0$, where $||x||_\infty$ denotes the infinity norm. 

The closed-loop dynamics under this nominal control are approximated using a ReLU neural network with a single hidden layer of eight neurons. 
The invariant set is estimated by solving the optimization problem~\eqref{eq:opt UIS} with multiple linear $\alpha$ values, specifically $\alpha = [0.025, 0.05, 0.06]$, and with no uncertainty included in the PWA approximation. 
As illustrated in Fig.~\ref{fig:UIS}, each level set corresponds to an invariant set obtained from the single-$\alpha$ method of~\cite{samanipour2024invariant}. 
The UIS max-composition certifies their union, yielding a strictly larger invariant set than any individual level set. 
Moreover, the certified set is verified on the original nonlinear dynamics, and the computational times are reported in Table~\ref{tab:times}.
\begin{figure}
    \centering
    \includegraphics[width=0.85\linewidth]{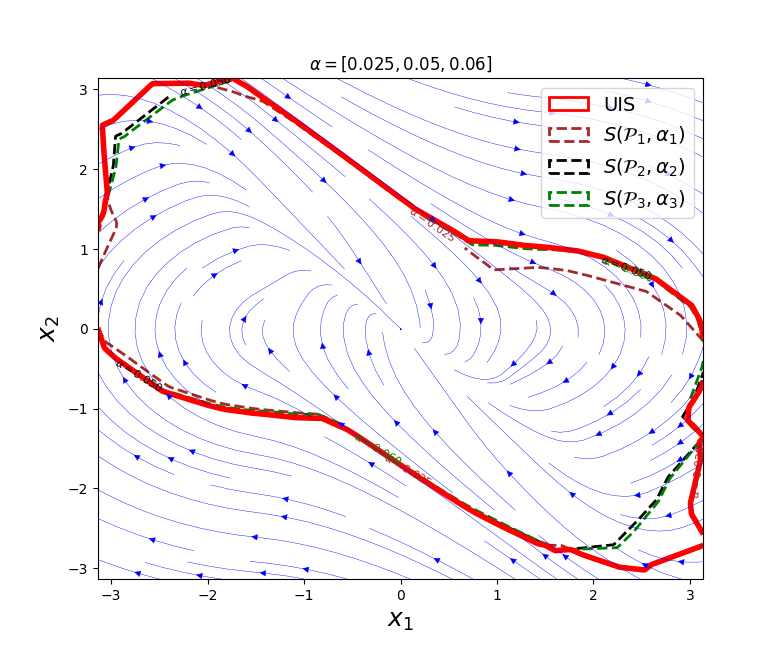}
    \caption{
    The final invariant set is obtained by verifying a CBF synthesized through the UIS method with three different choices of $\alpha$, as described in Sec.~\ref{sec:UIS}.
    }
    \label{fig:UIS}
\end{figure}

\end{example}
\begin{example}[Barrier Certificate Verification]\label{example:comparison example}
We consider the nonlinear system from~\cite{zhao2021synthesizing,zhang2023exact}:
\[
\dot{x}_1 = x_1^2 + x_1 x_2 + x_1, \quad
\dot{x}_2 = x_1 x_2 + x_2^2 + x_2,
\]
over the domain $\mathcal{D} = \{ x : \|x\|_\infty \le 2 \}$. The goal is to verify that trajectories starting in the initial region $I_m = \{ 0.5 \le x_1, x_2 \le 1 \}$ do not enter the unsafe set $U_m = \{ -1 \le x_1, x_2 \le 0 \}$.
We approximate the system using a single-hidden-layer ReLU neural network with 20 neurons. Invariant sets corresponding to fixed $\alpha$ values are computed via~\eqref{eq:opt UIS}, and their max-composition is certified using the UIS approach. This results in a provably safe set that is strictly larger than any single-$\alpha$ solution or SOS-based barrier, and comparable in coverage to SyntheBC~\cite{zhao2021synthesizing}.

Table~\ref{tab:times} reports the total certification time. While~\cite{zhang2023exact} uses a deeper network (two hidden layers with 256 ReLU neurons) and spends over 315 seconds (or several hours using SMT) to certify a Darboux-style barrier, our approach uses a compact single-layer ReLU model with only 20 neurons. A comparison of the results is shown in Fig.~\ref{fig:path-following}.
\end{example}

\begin{example}[Cart-Pole with Input Saturation]
\label{example:cartpole}
We study the classic cart-pole system~\cite{tedrake2009underactuated} with state 
$x, \dot{x}, \theta, \dot{\theta}$, cart mass $m_c=1.0$, pole mass $m_p=0.1$, pole length $l=1.0$, and gravity $g=9.81$. 
The horizontal input force $f_x$ is saturated within $[-30,30]$ N. 
Following~\cite{tedrake2009underactuated}, we shift $\theta \mapsto \pi-\theta$ so that the upright equilibrium is at the origin. 

The system is stabilized by the LQR controller from~\cite{wu2023neural}, 
\[
u = x + 2.4109\,\dot{x} + 34.3620\,\theta + 10.7009\,\dot{\theta},
\]
with domain $\mathcal{D} = \{x \in \mathbb{R}^4 : 1-|x|_{\infty} \geq 0\}$. 
The closed-loop dynamics are identified by a single-hidden-layer ReLU NN (15 neurons).
Applying the UIS method with two slopes $\alpha_1=0.35$ and $\alpha_2=0.55$ yields a certified invariant set, 
demonstrating the scalability of our framework to 4D nonlinear dynamics with input saturation. 
The certification is verified on the nonlinear dynamics, and the computational times are summarized in Table~\ref{tab:times}.
\end{example}

\begin{table}[t]
\centering
\caption{Computation and verification times (in seconds) for different examples.}
\label{tab:times}
\begin{tabular}{lccc}
\toprule
\textbf{Example} & \textbf{UIS Computation} & \textbf{Verification} & \textbf{Total} \\
\midrule
Inverted Pendulum~\ref{example:IP}   & 34   & 12   & 46 \\
Barrier Certificate~\ref{example:comparison example} & 57   & 23   & 80 \\
Cart-Pole~\ref{example:cartpole}           & 1349 & 117  & 1466 \\
\bottomrule
\end{tabular}
\end{table}
\section{Conclusion}
We presented an approximation--verification framework for synthesizing CBFs in nonlinear systems with input constraints. The approach combines convex synthesis on PWA surrogates using a two-slope Leaky ReLU $\alpha(\cdot)$, enlargement via a Union of Invariant Sets (UIS), and facet-wise verification with local uncertainty updates. 
Simulations on pendulum and cart-pole systems show that the method certifies larger invariant sets than linear-$\alpha$ designs while remaining tractable. Future work will address scalability to higher-dimensional systems and integration with predictive control.

\begin{figure}
    \centering
    \includegraphics[width=0.87\linewidth]{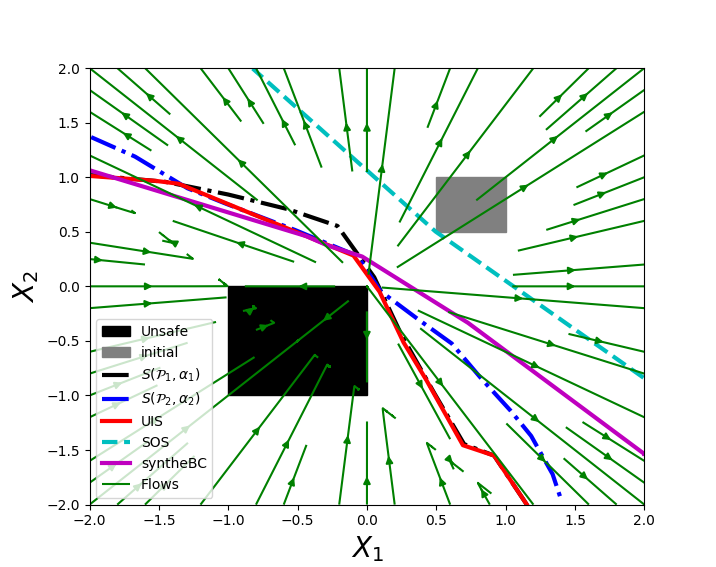}
    \caption{
    The final invariant set is obtained by UIS with two different $\alpha$ values (see Sec.~\ref{sec:UIS}) for Example~\ref{example:comparison example}. 
The certified region is compared against baselines from sum-of-squares (SOS) programming and SyntheBC, showing that UIS achieves a larger invariant set.}
    \label{fig:path-following}
\end{figure}
\bibliographystyle{ieeetr}
\bibliography{acc2025}
\end{document}